\documentclass[10pt,conference]{IEEEtran}

\usepackage[utf8x]{inputenc}
\usepackage{amsfonts}
\usepackage{amssymb}
\usepackage{amsthm}
\usepackage{amsmath}
\usepackage{hyperref}
\usepackage{multirow}
\usepackage{balance}
\usepackage{tikz}
\usetikzlibrary{automata, arrows}
\usepackage{psfrag}

\newcommand{\lc}{\left\lceil}
\newcommand{\rc}{\right\rceil}

\newtheorem{theorem}{Theorem}
\newtheorem{lemma}{Lemma}

\newcommand{\A}{\mathcal{A}}
\newcommand{\post}{\mathit{Post}}
\newcommand{\steps}[0]{ {\rightarrow} }

\newtheorem{definition}{Definition}

\title{Parametric Schedulability Analysis of Fixed Priority Real-Time
  Distributed Systems}

\author{ \IEEEauthorblockN{Youcheng Sun\IEEEauthorrefmark{1}, 
    Romain Soulat\IEEEauthorrefmark{2}, 
    Giuseppe Lipari\IEEEauthorrefmark{3}, 
    Étienne André\IEEEauthorrefmark{4}, 
    Laurent Fribourg\IEEEauthorrefmark{2}}
  \IEEEauthorblockA{\IEEEauthorrefmark{1}Scuola Superiore Sant'Anna,
    Italy\\y.sun@sssup.it}
  \IEEEauthorblockA{\IEEEauthorrefmark{3}LSV - ENS Cachan, France\\
    g.lipari@lsv.ens-cachan.fr}
  \IEEEauthorblockA{\IEEEauthorrefmark{4}LIPN, CNRS UMR 7030,
    Université Paris 13, France\\Etienne.Andre@univ-paris13.fr} 
  \IEEEauthorblockA{\IEEEauthorrefmark{2}LSV - ENS Cachan and CNRS,
    France\\\{soulat,fribourg\}@lsv.ens-cachan.fr}}

\begin{document}

\maketitle

\begin{abstract}
  Parametric analysis is a powerful tool for designing modern embedded
  systems, because it permits to explore the space of design
  parameters, and to check the robustness of the system with respect
  to variations of some uncontrollable variable.  In this paper, we
  address the problem of parametric schedulability analysis of
  distributed real-time systems scheduled by fixed priority. In
  particular, we propose two different approaches to parametric
  analysis: the first one is a novel technique based on classical
  schedulability analysis, whereas the second approach is based on
  model checking of Parametric Timed Automata (PTA).
  
  The proposed analytic method extends existing sensitivity analysis
  for single processors to the case of a distributed system,
  supporting preemptive and non-preemptive scheduling, jitters and
  unconstrained deadlines. Parametric Timed Automata are used to model
  all possible behaviours of a distributed system, and therefore it is
  a necessary and sufficient analysis. Both techniques have been
  implemented in two software tools, and they have been compared with
  classical holistic analysis on two meaningful test cases. The
  results show that the analytic method provides results similar to
  classical holistic analysis in a very efficient way, whereas the PTA
  approach is slower but covers the entire space of solutions.
\end{abstract}

\section{Introduction and motivation}
\label{sec:intr-motiv}

Designing and analysing a distributed real-time system is a very
challenging task. The main source of complexity arises from the large
number of parameters to consider: task priority, computation times and
deadlines, synchronisation, precedence and communication constraints,
etc. Finding the ``optimal'' values for the parameters is not easy,
and often the robustness of the solution strongly depends on the exact
values: a small change in one parameter may completely change the
behaviour of the system and even compromise the correctness. For these
reasons, designers are looking for analysis methodologies that enable
incremental design and exploration of the space of parameters.

Task computation times are particularly important parameters. In
modern processor architectures, it is very difficult to precisely
compute worst-case computation times of tasks, and estimations derived
by previous executions are often used in the analysis. However,
estimations may turn out to be optimistic, hence an error in the
estimation of a worst-case execution time may compromise the
schedulability of the system.

The goal of this research is to characterise the space of the
parameters of a real-time system for which the system is schedulable,
i.e. all tasks meet their deadlines. Parametric analyses for real-time
systems have been proposed in the past, especially on single
processors \cite{Bini-thesis,Bin07b,Palopoli-rtss08,AFKS12}.

In this paper, we investigate the problem of doing parametric analysis
of real-time distributed systems scheduled by fixed priority. We
consider an application modelled by a set of pipelines of tasks (also
called \emph{transactions} in \cite{Tindell:1994:HSA:195612.195618}),
where each pipeline is a sequence of tasks that can be periodic or
sporadic, and all tasks in a pipeline must complete before an
end-to-end deadline. We consider that all nodes in the distributed
system are connected by one or more CAN bus \cite{Davis07}.

We propose:
\begin{itemize}
\item a new method for doing parametric analysis of distributed
  real-time systems scheduled by fixed priority scheduling. The method
  extends the sensitivity analysis proposed by Bini et
  al. \cite{Bin07b,Bini-thesis} by considering distributed systems and
  non-preemptive scheduling.
\item a model of a distributed real-time system using parametric timed
  Automata, and a model checking methodology using the Inverse Method
  \cite{AS13,AFKS12,FLMS-time12};
\item comparison of these two approaches with classical holistic
  analysis using the MAST tool \cite{MAST-2001,MAST-web-page}, in
  terms of complexity and precision of the analysis.
\end{itemize}

\section{Related Work}
\label{sec:rel-work}

There has been a lot of research work on parametric schedulability
analysis, especially on single processor systems. Bini and Buttazzo
\cite{Bin04b} proposed an analysis of fixed priority single processor
systems based on Lehoczky test \cite{Leh89}. Later, Bini, Di Natale
and Buttazzo \cite{Bin07b} proposed a more complex analysis, which
considers also the task periods as parameters. Such results are
summarised and extended in Bini's PhD thesis \cite{Bini-thesis}.

Parameter sensitivity can be also be carried out by repeatedly
applying classical schedulability tests, like the holistic analysis
\cite{Palencia1998,Tindell:1994:HSA:195612.195618}. One example of
this approach is used in the MAST tool \cite{MAST-2001,MAST-web-page},
in which it is possible to compute the \emph{slack} (i.e. the
percentage of variation) with respect to one parameter for single
processor and for distributed systems by applying binary search in
that parameter space \cite{Palencia1998}.

A similar approach is followed by the SymTA/S tool \cite{Symtas-05},
which is based on the \emph{event-stream} model
\cite{Richter:2002:EMI:882452.874327}. Another interesting approach is
the Modular Performance Analysis (MPA)
\cite{Wandeler:2006:SAE:1177177.1177184} which is based on Real-Time
Calculus \cite{thiele2000real}. In both cases, the analysis is
compositional, therefore less complex than the holistic analysis;
nevertheless, these approaches are not fully parametric, in the sense
that it is necessary to repeat the analysis for every combination of
parameters values in order to obtain the schedulability region.

Model checking on \emph{Parametric Timed Automata} (PTA) can be used
for parametric schedulability analysis, as proposed by Cimatti,
Palopoli and Ramadian \cite{Palopoli-rtss08}. In particular, thanks to
generality of the PTA modelling language, it is possible to model a
larger class of constraints, and perform parametric analysis on many
different variables, for example task offsets. Their approach has been
recently extended to distributed real-time systems \cite{Pal13}.

Also based on PTA is the approach proposed by André et
al. \cite{AFKS12}. Their work is based on the Inverse Method
\cite{AS13} and it is very general because it permits to perform
analysis on any system parameter. However, this generality can be paid
in terms of complexity.

In this paper, we first propose an extensions of the methods in
\cite{Bini-thesis} for distributed real-time systems. We also propose
a model of a distributed real-time systems in PTA, and compare the two
approaches against classical holistic analysis.

\section{System model}
\label{sec:system-model}

We consider distributed real-time systems consisting of several
computational nodes, each one hosting one single processor, connected
by one or more shared networks. We consider preemptive fixed priority
scheduling for processors, as this is the most popular scheduling
algorithm used in industry today, and non-preemptive fixed priority
scheduling for networks. In particular, the CAN bus protocol is a very
popular network protocol that can be analysed using non-preemptive
fixed priority scheduling analysis \cite{Davis07}. We will consider
extensions to our methodology to other scheduling algorithms and
protocols in future works.

For the sake of simplicity and uniformity of notation, in this paper
we use the same terminology to denote processors and communication
networks, and tasks and messages. Therefore, without loss of
generality, from now on we will use the term \emph{task} to denote
both tasks and messages, and the term \emph{processor} to denote both
processors and networks

A distributed real-time system consists of a set of task pipelines
$\{\mathcal{P}^1, \ldots, \mathcal{P}^n\}$ to be executed on a set of
$m$ processors $\{p_1, p_2, \ldots, p_m\}$. In order to simplify the
notation, in the following we sometime drop the pipeline index when
there is no possibility of misinterpretation.

A \emph{pipeline} is a chain of tasks $\mathcal{P} = \{\tau_1, \ldots,
\tau_n\}$, and each task is allocated on one possibly different
processor. A pipeline is assigned two fixed parameters: $T$ is the
pipeline period, and $D_{e2e}$ is the end-to-end deadline. This means
that the first task of the pipeline is activated every $T$ units of
time, and every activation is an \emph{instance} (or \emph{job}) of
the task. We denote the $k$-th instance of task $\tau_i$ as
$\tau_{i,k}$. Every successive task in the pipeline is activated when
the corresponding instance of the previous task has completed;
finally, the last task must complete before $D_{e2e}$ units of time
from the activation of the first task. Therefore, tasks must be
executed in a sequence: job $\tau_{i,k}$ cannot start executing before
job $\tau_{i-1,k}$ has completed.

A task can be a piece of code to be executed on a CPU, or a message to
be sent on a network. More precisely, a real-time periodic task
$\tau_i = (C_i, T_i, D_i, \overline{D}_i, \pi_i, J_i)$ is modelled by
the following fixed parameters:
\begin{itemize}
\item $T_i$ is the task period. All tasks in the same pipeline have
  period equal to the pipeline period $T$;
\item $\pi_i$ is the task priority; the higher $\pi_i$, the larger the
  priority;
\item $\overline{D}_i$ is the task \emph{fixed deadline}; all jobs of
  $\tau_i$ must complete within $\overline{D}_i$ from their
  activation.
\end{itemize}
Also, a task has the following free parameters:
\begin{itemize}
\item $C_i$ is the worst-case computation time (or worst-case
  transmission time, in case it models a message). In this paper we
  want to characterise the schedulability of the system in the space
  of the computation times, so $C_i$ is a free parameter.
\item $D_i$ is a variable denoting an upper bound on the task
  worst-case completion time. We will call this variable \emph{actual
    task deadline} or simply \emph{task deadline}. Of course, we
  require that $D_i \leq \overline{D}_i$. Remember that fixed priority
  does not use the task deadline for scheduling, but just for
  schedulability analysis. As we will see later, we will use this
  variable for imposing precedence constraints on pipelines. We say
  that a task has constrained deadline when $D_i \leq T_i$, and
  unconstrained deadline when $D_i > T_i$.
\item $J_i$ is the task start time jitter (see below).
\end{itemize}

As anticipated, a task consists of an infinite number of jobs
$\tau_{i,k}, k=1,\ldots$. Each job is activated at time $a_{i,k} =
kT_i$, can start executing (or can be sent on the network) no earlier
than time $s_{i,k}$, with $a_{i,k} \leq s_{i,k} \leq a_{i,k} + J_i$,
executes (or is transmitted over the network) for $c_{i,k} \leq C_i$
units of time, and completes (or is received) at $f_{i,k}$. For the
task to be schedulable, it must be $\forall k, f_{i,k} \leq d_{i,k} =
a_{i,k} + D_i$. A sporadic task has the same parameters as a periodic
task, but parameter $T_i$ denotes the \emph{minimum inter-arrival time}
between two consecutive instances. We also define the $i$-th level
hyperperiod as $H_i = \mathsf{lcm}(T_1, \ldots, T_i)$.

In this paper, we use the following convention. All tasks belonging to
a pipeline $\mathcal{P} = \{\tau_1, \ldots, \tau_n\}$ are activated at
the same time $a_{i,k} = a_{1,k}$. However, only the first task can
start executing immediately: $s_{1,k} = a_{1,k}$. The following tasks
can only start executing when the previous task has completed:
$\forall i=2,\ldots,n \;\; s_{i,k} = f_{i-1, k}$. The task jitter is
the worst case start time of a task: $J_i \geq \max_k \{s_{i,k} -
a_{i,k}\}$.



A scheduling algorithm is \emph{fully preemptive} if the execution of
a lower priority job can be suspended at any instant by the arrival of
a higher priority job, which is then executed in its place. A
scheduling algorithm is \emph{non-preemptive} if a lower priority job
can complete its execution regardless of the arrival of higher
priority jobs. In this paper, we consider preemptive fixed priority
scheduling for CPUs, and non-preemptive fixed priority scheduling for
networks.

\section{Analytic method}
\label{sec:analytic}

In this section we describe a novel method for parametric analysis of
distributed system. The method is based on the sensitivity analysis by
Bini et al. \cite{Bin07b,Bini-thesis}, and extends it to include
jitter and deadline parameters.

\subsection{Single processor preemptive fixed priority scheduling}

There are many ways to test the schedulability of a set of real-time
periodic tasks scheduled by fixed priority on a single processor. In
the following, we will use the test proposed by Seto et
al. \cite{seto98} because it is amenable to parametric analysis of
computation times, jitters and deadlines.

With respect to the original formulation, we now consider tasks with
constrained deadlines (i.e. $D_i$ can be less than or equal to $T_i$).
\begin{theorem}
  \label{th:setos-test}
  Consider a system of sporadic tasks $\{\tau_1, \ldots, \tau_n\}$
  with constrained deadlines and zero jitter, executed on a single
  processor by a fixed priority scheduler. Assume all tasks are
  ordered in decreasing order of priorities, with $\tau_1$ being the
  highest priority task.

  Task $\tau_i$ is schedulable if and only if:
  \begin{equation}
    \label{eq:seto}
    \exists \mathbf{n} \in \mathbb{N}^{i-1}  \left \{ 
      \begin{array}{ll}
        C_i + \sum_{j=1}^{i-1} n_j C_j \leq n_k T_k  & \forall k = 1, \ldots, i-1\\
        C_i + \sum_{j=1}^{i-1} n_j C_j \leq D_i
      \end{array} \right.
  \end{equation}
  where $\mathbb{N}^{i-1}$ is the set of all possible vectors of $(i-1)$
  positive integers.
\end{theorem}
\begin{proof}
  See \cite{Bini-thesis} and \cite{seto98}.
\end{proof}
Notice that, with respect to the original formulation, we have
separated the case of $k=i$ from the rest of the inequalities.

The theorem allows us to only consider sets of linear inequalities,
because the non-linearity has been encoded in the variables $n_j$. The
resulting system is a set of inequalities in disjunctive and
conjunctive form. Geometrically, this corresponds to a non-convex
polyhedron in the space of the variables $C_i, D_i$.

How many vectors $\mathbf{n}$ do we have to consider? If the deadline
$D_i$ is known, the answer is to simply consider all vectors
corresponding to the minimal set of scheduling points by Bini and
Buttazzo \cite{Bini-2004-ID20}.  If $D_i$ is unknown, we have to
consider many more vectors: more specifically, we must select all
multiples of the period of any task $\tau_j$ with priority higher than
$\tau_i$, until the maximum possible value of the deadline. All
vectors until time $t$ can be computed as:
\begin{equation}
\label{eq:number_points}
\mathcal{B}^{i-1}(t) = \left \{ \mathbf{n} \;| \;\exists k, h, kT_h \leq
  t : \forall j, \; n_j = \lc \frac{kT_h}{T_j} \rc \right \}.
\end{equation}
If a task is part of a pipeline with end-to-end deadline equal to
$D_{e2e}$, then $D_i \leq D_{e2e}$ (keep in mind that, by now, the
deadline is supposed to not exceed the task period). Therefore, we
have to check all $\mathbf{n} \in \mathcal{B}^{i-1}(D_{e2e})$.

The number of vectors (and correspondingly, the number of
inequalities) depends on the relationship between the task periods. In
real applications, we expect the periods to have ``nice''
relationships: for example, in many cases engineers choose periods
that are multiples of each others. Therefore, we expect the set of
inequalities to have manageable size for realistic problems.

We have one such non-convex region for every task $\tau_i$.  Since we
have to check the schedulability of all tasks on a CPU, we must
\emph{intersect} all such regions to obtain the final region of
schedulable parameters.

\subsection{Unconstrained deadlines and jitters}

We now extend Seto's test to unconstrained deadlines and variable
jitters. When considering a task with deadline greater than period,
the worst-case response time may be found in any instance, not
necessarily in the first one (as with the classical model of
constrained deadline tasks). Therefore, we have to check the workload
not only of the first job, but also of the following jobs of
$\tau_i$. Let use define $h_i = \frac{H_i}{T_i}$, i.e. the number of
jobs of $\tau_i$ contained in the $i$-level hyperperiod.  Then, task
$\tau_i$ is schedulable if and only if the following system of
inequalities is verified:
\begin{align}
  \label{eq:h-jobs}
  \forall& h = 1, \ldots, h_i, \exists \mathbf{n} \in
  \mathbb{B}^{i-1}(hT_i + D_{e2e}) \\ \nonumber
  & \left \{
    \begin{aligned}
      h C_i &+ \sum_{j=1}^{i-1} n_j C_j \leq n_k T_k, \forall k=1,\ldots,i-1 \\
      h C_i &+ \sum_{j=1}^{i-1} n_j C_j \leq (h-1) T_i + D_i
    \end{aligned} \right.
\end{align}

The correctness of the test is proved by the following Lemma.

\begin{lemma}
  \label{lm:D-g-T-equivalence}
  Consider a system $\mathcal{T} = \{\tau_1, \ldots, \tau_{i-1},
  \tau_i\}$. Let $\mathcal{T}^{(h)}$ be a task set obtained from
  $\mathcal{T}$ by substituting $\tau_i$ with $\tau_i^{(h)}$ having
  computation time $C_i^{(h)} = h C_i$, deadline $D_i^{(h)} =
  (h-1) T_i + D_i$ and the same priority $\pi_i^{(h)} = \pi_i$.

  If for every $h = 1, \ldots, h_i$, task $\tau_i^{(h)}$
  completes before its deadline, then the first $h_i$ jobs of
  $\tau_i$ will also complete before their deadlines.
\end{lemma}
\begin{proof}
  By induction. Base of induction: the response time of job $h=1$
  corresponds to the response time of the first job of $\tau_i^{(1)}$ (trivially
  true). Therefore, if $\tau_i^{(1)}$ is schedulable, also the first
  job of $\tau_i$ is schedulable.

  Now, the induction step.  Suppose the Lemma is valid for
  $h=1,\ldots, k$, we are now going to prove that is also valid for
  $h=1,\ldots,k+1$. By assumption, the first job of $\tau_i^{(h)}$ is
  schedulable for $h=1,\ldots,k$. As a consequence of the validity of
  the Lemma, also the first $k$ instances of $\tau_i$ are
  schedulable. Let $f_{i,k}$ be the finishing time of the first job of
  $\tau_i^{(k)}$. We have two cases: either $f_{i,k} \leq kT_i$, or
  $kT_i < f_{i,k} \leq (k-1)T_i + D_i$. 

  In the first case, job $k+1$ is only subject to the interference of
  higher priority tasks. Therefore, its worst case response time
  correspond to the situation in which all higher priority tasks
  arrive at the same time $kT_i$ (critical instant), and it is
  therefore equal to the response time of the first job $h=1$, hence
  also schedulable. We can conclude that the Lemma is true without
  further induction steps.

  In the second case, the $k+1$ job has to wait for the previous job
  $k$ to finish before it can start executing. In particular, there is
  no idle time in interval $[0, f_{i,k+1}]$. Therefore, the response
  time of job $k+1$ coincides with the response time of task
  $\tau_i^{(k+1}$, and if the second one is schedulable, also job
  $k+1$ is schedulable.

  Finally, since the first instance of $\tau_i^{(h)}$ is schedulable
  for all $h=1, \ldots, h_i$, and given that $C_i^{(h)} = h C_i$, then
  from Theorem \ref{th:setos-test} follows that the system of
  Inequalities in (\ref{eq:h-jobs}) is verified.
\end{proof}

To take into account the task jitter, we can appropriately adjust the
last term that accounts for the task deadline, and the set
$\mathcal{B}^{i-1}(t)$.
\begin{theorem}\label{thm:jitter}
  Task $\tau_i$ is schedulable if:
  \begin{align}
    \label{eq:final-single-proc}
    \forall &h = 1, \ldots, \frac{H_i}{T_i}, \;\; \exists \mathbf{n} \in \mathbb{B}^{i-1}(D_{e2e}) \;\; \\
    &\left \{
      \begin{aligned}
        h C_i &+ \sum_{j=1}^{i-1} n_j C_j \leq n_k T_k - J_k \; \forall k = 1, \ldots, i-1\\
        h C_i &+ \sum_{j=1}^{i-1} n_j C_j \leq (h - 1) T_i + D_i - J_i        
      \end{aligned} \right.
  \end{align}
  where 
  \[
  \mathcal{B}^{i-1}(t) = \left \{ \mathbf{n} \;| \; \exists k, h, kT_h-\overline{D}_h \leq
    t : \forall j n_j = \lc \frac{kT_h+\overline{D_h}}{T_j} \rc \right \}.
  \]
\end{theorem}
\begin{proof}
  We report here a sketch of the complete proof. For every higher
  priority interfering task $\tau_k$, the worst case situation is when
  the first instance arrives at $J_k$, whereas the following instances
  arrive as soon as it is possible. Therefore, the scheduling points
  must be modified from $n_k T_k$ to $n_k T_k - J_k$. For what
  concerns task $\tau_i$, the critical instant corresponds to the
  situation in which the first instance can only start at $J_i$, hence
  the available interval is $(h - 1) T_i + D_i - J_i$.
\end{proof}

Notice that the introduction of unconstrained deadline adds a great
amount of complexity to the problem. In particular, the number of
non-convex regions to intersect is now $\mathcal{O}(\sum_{i=1}^n
\frac{H_i}{T_i})$, which is dominated by $\mathcal{O}(n H_n)$. So, the
proposed problem representation does not scale with increasing
hyperperiods; however, as we will show in Section
\ref{sec:experiments}, the problem is tractable when periods are
harmonic or quasi-harmonic, as it often happens in real applications.

\subsection{Non preemptive scheduling}
\label{sec:non-preemptive}

In this paper we model the network as a non-preemptive fixed priority
scheduled resource. 
In non-preemptive fixed priority scheduling, the worst-case response
time for a task $\tau_{i}$ can be found in its longest $i$-level
active period~\cite{Bril:07}. A $i$-level active period $L_i$ is an
interval $[a, b)$ such that the amount of processing that needs to be
performed due to jobs with priority higher than or equal to $\tau_{i}$
(including $\tau_i$ itself) is larger than 0 $\forall t \in (a, b)$,
and equal to 0 at instants $a$ and $b$. The longest $L_i$ can be found
by computing the lowest fixed point of the following recursive
function~\cite{George:96}:
\begin{align}\label{eq:level-i-active-period}
  &\left \{
    \begin{aligned}
      L_{i}^{0} &= B_{i} + C_{i} \\
      L_{i}^{(s)} &= B_{i} + \sum_{j <= i}\lceil \frac{L_{i}^{(s-1)}}{T_{j}} \rceil C_{j} 
    \end{aligned} \right. 
\end{align}
\noindent where $B_i = \max_{i < j}(C_i - 1)$.

In order to find the worst-case response time of task $\tau_i$, all
jobs $\tau_{i,k}$ that appear in the longest $L_i$ need to be checked,
with $ k \in [1, \lceil\frac{L_i}{T_i}\rceil]$.

To obtain the worst-case response time, we compute first its
worst-case start time. When there is no jitter, George et
al. \cite{George:96} give the following formula to compute the
worst-case start time of a job $\tau_{i, k}$ :
\begin{align}\label{rq:start_time}
 &\left \{
  \begin{aligned}
    s_{i,k}^{(0)} &= B_{i} + \sum_{j < i}C_{j}\\
    s_{i,k}^{(l+1)} &= B_{i} + (k-1)C_{i} + \sum_{j < i}
    (\left \lfloor \frac{s_{i,k}^{l}}{T_{j}} \right\rfloor + 1) C_{j}\\
  \end{aligned}\right.
\end{align}
\noindent Note that $(k-1)C_i$ is the computation time of the
preceding $(k-1)$ jobs. Since a lower priority task's execution cannot
be preempted, this could ``push'' one job of a higher priority task to
interfere with its future jobs.

Observe that the iterating computation of $L_{i}$ in
Equation~(\ref{eq:level-i-active-period}) is non decreasing and (when
the system utilisation is no larger than 1) $B_{i} + \sum_{j<=i}
\lceil \frac{H_i}{T_i} \rceil C_i <= B_i +H_i$, so the length of $L_i$
will not exceed $B_i+H_i$.

In this paper, the worst-case execution time of the tasks are
considered free parameters. However, $L_i$ can still be upper bounded
by $\overline{L_i} = \max_{i<j}(T_j) + H_i$. Now, we can derive a
similar feasibility test for non preemptive scheduling as in
Theorem~\ref{thm:jitter}.
\begin{theorem}\label{thm:jitter-np}
  A non preemptive task $\tau_i$ is schedulable if :
  \begin{align}\label{eq:final-single-proc-np}
    &\forall h = 1, \ldots, \lceil\frac{\overline{L_i}}{T_i}\rceil, \;\; \exists \mathbf{n} \in \mathbb{B}^{i-1}(D_{e2e}) \;\;\\
    &\left \{
      \begin{aligned}
        B_{i} &+ (h-1)C_{i} + \sum\limits_{j=1}^{i-1}n_{j}C_{j} 
        \leq n_{l}T_{l} - J_l \;\; \forall l = 1, \ldots, i-1 \\
        B_{i} &+ (h-1)C_{i} + \sum\limits_{j=1}^{i-1}n_{j}C_{j} 
        \le (h-1)T_i + D_{i} - C_{i} - J_i \\
      \end{aligned}\right.
  \end{align}
  where $\mathbb{B}^{i-1}(D_{e2e})$ is defined as in
  Theorem~\ref{thm:jitter}, and $B_i$ is the blocking time that task
  $\tau_i$ suffers from lower priority tasks:
  \[
  \forall i, \forall j>i \;\;  B_i \leq C_j - 1
  \]
\end{theorem}
\begin{proof}
  See the sufficient part of proof in~\cite{seto98} and
  Theorem~\ref{thm:jitter}.
\end{proof}
Term $B_i$ is an additional free variable used to model the blocking
time that a task suffers from lower priority tasks. It is possible to
avoid the introduction of this additional variable by substituting it
in the inequalities with a simple Fourier-Motzkin elimination.

Like in the preemptive case, for every non preemptive task, this
theorem builds a set of inequalities. The system schedulability region
is the intersection of all the sets. The complexity of this procedure
is the same as for the preemptive case.
 
\subsection{Distributed systems}
\label{sec:analysis-distr}

Until now, we have considered the parametric analysis of independent
tasks on single processor systems, with computation times, deadlines
and jitter as free parameters. In particular, the equations in Theorem
\ref{thm:jitter} and Theorem \ref{thm:jitter-np} give us a way to
express the constraints on the system in a fully parametric way: all
solutions to the system of Inequalities \eqref{eq:final-single-proc}
and \eqref{eq:final-single-proc-np} are all the combinations of
computations times, deadlines and jitters that make the single
processor system schedulable.

It is important to make one key observation. If we fix the computation
times and the jitters of all tasks, and we leave the deadlines as the
only free variables, the worst-case response time of each task can be
found by minimising the deadline variables. As an example, consider
the following task set (the same as in \cite{Bin04b}) to be scheduled
by preemptive fixed priority scheduling on a single processor:

\begin{center}
  \begin{tabular}{|l|c|c|c|c|} \hline Task & $C_i$ & $T_i$ & $D_i$ &
    $p_i$ \\\hline
    $\tau_1$ &   1   &  3    &  3    &  3    \\
    $\tau_2$ &   2   &  8    &  7    &  2    \\
    $\tau_3$ & 4 & 20 & ?  & 1 \\\hline
  \end{tabular}
\end{center}

We consider $D_3$ as a parameter and set up the system of inequalities
according to Equation \eqref{eq:final-single-proc}. After reduction of
the non-useful constraints, we obtain
\[
   12 \leq D_3 \leq 20
\]
Notice that $12$ is actually the worst-case response time of $\tau_3$.

The second key observation is that a precedence constraint between two
consecutive tasks $\tau_i$ and $\tau_{i+1}$ in the same pipeline can
be expressed as $ D_i \leq J_{i+1}$. This basically means that the
worst-case response time of task $\tau_i$ should never exceed the
jitter (i.e. worst-case start time) of task $\tau_{i+1}$. Therefore,
we have a way to relate tasks allocated on different processors that
belong to the same pipeline. 

Finally, the last task in every pipeline, let us call it $\tau_n$ must
complete before the end-to-end deadline: $D_n \leq D_{e2e}$.

We are now ready use inequalities in \eqref{eq:final-single-proc} as
building blocks for the parametric analysis of distributed
systems. The procedure to build the final system of inequalities is as
follows:
\begin{enumerate}
\item For each processor, we build the system of inequalities
  \eqref{eq:final-single-proc}, and for every network the system of
  inequalities in \eqref{eq:final-single-proc-np}. All these systems
  are independent of each other, because they are constraints on
  different tasks, so they use different variables. The combined
  system contains $3*N$ variables, where $N$ is the total number of
  tasks.
\item For every pipeline, we add the following precedence constraints:
  \begin{itemize}
  \item For the first task in the pipeline, let us denote it as
    $\tau_1$, we set its jitter to 0: $J_1 = 0$.
  \item For every pair of consecutive tasks, let us denote them as
    $\tau_i$ and $\tau_{i+1}$, we impose the precedence constraint:
    $D_i \leq J_{i+1}$;
  \item For the last task in the pipeline, let us denote it as
    $\tau_n$, we impose that it must complete before its end-to-end
    deadline $D_n \leq D_{e2e}$.
  \end{itemize}
  Such constraints must intersect the combined system to produce the
  final system of constraints.
\end{enumerate}

To give readers an idea how the parameter space of a distributed
system would look like, here is a very simple example, built with the
goal of showing the general methodology without taking too much space.
We consider a system with two processors (and no network), two tasks
$\tau_1$ and $\tau_3$, and one pipeline consisting of two tasks,
$\tau_{21}$ and $\tau_{22}$.
\begin{center}
  \begin{tabular}{|c|ccccc|} \hline Pipeline & Task & $\pi_i$ & Resource 
	& $T_i$ & $\overline{D}_{i} (D_{e2e})$ \\\hline
    	- & $\tau_1$ &  2  & CPU1 & 10 & 4 \\\hline
	\multirow{2}{*}{$P^2$} 
	  & $\tau_{1}^2$ &  1 & CPU1  & \multirow{2}{*}{20} &\multirow{2}{*}{6}\\
	  & $\tau_{2}^2$ &  2 & CPU2  & 		    	   & \\\hline
    	- & $\tau_3$ &  1  & CPU2 & 16 & 16 \\\hline
  \end{tabular}
\end{center}
To make sure that for each task we have one single inequality (see
Equation \eqref{eq:final-single-proc}) we set up the deadlines short
enough so that one schedulability point for each task needs to be
considered, thus avoiding complex disjoints.

Based on the analysis in this section, we derive a set of constraints,
where $J$, $C$ and $D$ are the free variables for the tasks.
\[
\begin{cases}
	J_1 \ge 0, C_1 \ge 0, C_{1}^2 \ge 0, C_{2}^2 \ge 0, J_3 \ge 0, C_{3} \ge 0\\
	D_1 \le 4, D_{3} \le 16 \\
	C_1 + J_{1} \le D_1 \\
	C_{21} + C_{1} \le D_{21} \\
	C_{2}^2 + J_{2}^2 \le D_{2}^2 \\
	C_{3} + C_{2}^2 + J_{2}^2 \le 20 \\
	C_{3} + C_{2}^2 + J_{3} \le D_{3} \\
	J_{1}^2=0, D_{1}^2 \le J_{2}^2, D_{2}^2 \le 6 \\
\end{cases}
\]
In the first two lines, we show the ``trivial'' inequalities: all
values must be non-negative, and every deadline must not exceed the
corresponding maximum deadline specified in the table.  The
inequalities at line 3 and 4 and the inequalities at line 5, 6 and 7
are (reduced) constraints (according to Theorem~\ref{thm:jitter}) on
the schedulability of tasks on processor 1 and 2,
respectively. Finally, the inequalities in the last line are the ones
imposed by the precedence constraints between $\tau_{1}^2$ and
$\tau_{2}^2$.

A real system will produce a much more complex set of constraints. For
each task we will need to prepare a set of disjoint inequalities, that
must be intersect with each other: this may greatly increment the
number of inequalities to be considered. Also, often we need to model
the network. Therefore, we prepared a software tool to automatically
build and analyse the set of inequalities for a distributed system.

\subsection{Implementation}

The analytic method proposed in this section has been implemented in
RTSCAN \cite{RTSCAN-web-page}, a C/C++ library publicly available as
open source code that collects different types of schedulability
tests. The code for the parametric schedulability analysis uses the
PPL (Parma Polyhedra Library) \cite{BagnaraHZ06TR}, a library
specifically designed and optimised to represent and operate on
polyhedra. The library efficiently operates on rational numbers with
arbitrary precision: therefore, in this work we make the assumption
that all variables (computations times, deadlines and jitter) are
defined in the domain of integers. This does not represent a great
problem, since in practice every value is multiple of a real-time
clock expressed as number of ticks.

An evaluation of this tool, and of the complexity of the analysis
presented here, will be presented in Section \ref{sec:experiments}.

\section{The Inverse Method approach}
\label{sec:timed-automata}

\subsection{Parametric Timed Automata} 

Timed Automata are finite-state automata augmented with clocks, i.e.,
real-valued variables increasing uniformly, that are compared within
guards and invariants with timing delays~\cite{AD94}.  Parametric
timed automata (PTAs)~\cite{AHV93} extend timed automata with
parameters, i.e., unknown constants, that can be used in guards and
invariants.

Formally, given a set~$X$ of clocks and a set~$P$ of parameters, a
constraint~$C$ over~$X$ and~$P$ is a conjunction of linear
inequalities on~$X$ and~$P$.  Given a parameter valuation (or
point)~$\pi$, we write~$\pi \models C$ when the constraint where all
parameters within~$C$ have been replaced by their value as in~$\pi$ is
satisfied by a non-empty set of clock valuations.

\begin{definition}
A PTA~$\A$ is \mbox{$(\Sigma, Q, q_{0}, X, P, K, I, \steps)$} with
	$\Sigma$ a finite set of actions,
	$Q$ a finite set of locations,
	$q_{0} \in Q$ the initial location,
	$X$ a set of clocks,
	$P$ a set of parameters,
	$K$ a constraint over~$P$,
	$I$ the invariant assigning to every $q\in Q$ a constraint over~$X$ and~$P$, and
	$\steps$ a step relation consisting of elements $(q,g,a,\rho,q')$, 
	where
	$q,q'\in Q$, $a\in\Sigma$, $\rho\subseteq X$ is the set of clocks to be reset, and
	the guard $g$ is a constraint over~$X$ and~$P$.
\end{definition}
 
The semantics of a PTA~$\A$ is defined in terms of states, i.e.,
couples $(q, C)$ where~$q \in Q$ and~$C$ is a constraint over~$X$
and~$P$.  Given a point~$\pi$, we say that a state $(q, C)$ is
$\pi$-compatible if $\pi \models C$.  Runs are alternating sequences
of states and actions, and traces are time-abstract runs, i.e.,
alternating sequences of \emph{locations} and actions.  The trace set
of~$\A$ corresponds to the traces associated with all the runs
of~$\A$.  Given~$\A$ and~$\pi$, we denote by~$\A[\pi]$ the
(non-parametric) timed automaton where each occurrence of a parameter
has been replaced by its constant value as in~$\pi$.
%
%
One defines $\post_{\A(K)}^i(S)$ as the set of states reachable from a
set~$S$ of states in exactly $i$~steps under~$K$,
and 
$\post_{\A(K)}^*(S)=\bigcup_{i\geq 0 }\post_{\A(K)}^i(S)$.

Detailed definitions on parametric timed automata can be found in,
e.g.,~\cite{AS13}.

The Inverse Method exploits the model of Timed Automata and the
knowledge of a reference point of timing values for which the good
behaviour of the system is known. The method synthesises automatically
a dense zone of points around the reference point, for which the
discrete behaviour of the system, that is the set of all the
admissible sequences of interleaving events, is guaranteed to be the
same.  Although the principle of the inverse method shares
similarities with sensitivity analysis, its algorithm proceeds by
iterative state space exploration. Furthermore, its result comes under
the form of a fully parametric constraint, in contrast to sensitivity
analysis.  By repeatedly applying the method, we are able to decompose
the parameter space into a covering set of ``tiles'', which ensure a
uniform behaviour of the system: it is sufficient to test only one
point of the tile in order to know whether or not the system behaves
correctly on the whole tile.

\subsection{System model with PTAs}

In this section, we show how we modelled a schedulability problem as
defined in \ref{sec:system-model}, similarly to what has been done in
\cite{FLMS-time12}. In the current implementation, we only model
pipelines with end-to-end deadlines no larger than their
periods. Moreover, all pipelines are strictly periodic, and have 0
offset. This means that the results of the parametric analysis
produced by this model are only valid for periodic synchronous
pipelines.

We illustrate our model with the help of an example of two pipelines
$\mathcal{P}^1,\mathcal{P}^2$ with $\mathcal{P}^1 =
\{\tau_1,\tau_2\}$, $\mathcal{P}^2 = \{\tau_3,\tau_4\}$, $p(\tau_1) =
p(\tau_4) = p_1$, $p(\tau_2) =p(\tau_3) = p_2$, $p_1$ being a
preemptive processor and $p_2$ being non-preemptive. We have that
$\pi_1 > \pi_4$ and $\pi_3 > \pi_2$.

In Figure \ref{fig:pipeline}, we show the model of a pipeline. A
pipeline is a sequence of tasks that are to be executed in order: when
a task completes its instance, it instantly activates the next one in
the pipeline. Once every task in the pipeline has completed, the
pipeline waits for the next period to start.

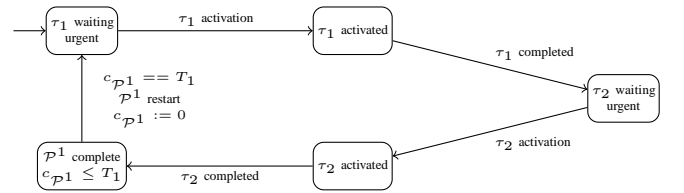
\begin{figure}[ht!]
{
  \centering
  \tiny
  
  \tikzstyle{state}=[rectangle, rounded corners, minimum size=18pt, draw=black, text=black, inner sep=1.5pt,initial text={}]
  \tikzstyle{point}=[draw=none, fill=black, circle]
  
  \begin{tikzpicture}[scale = 1.8, auto, ->, thin] 
    
    \node[initial, state] at (-2,0.5) (0) {\begin{tabular}{@{} c @{} }$\tau_1$ waiting \\ urgent\end{tabular}};
    \node[state] at (0,0.5) (1) {\begin{tabular}{@{} c @{} }$\tau_1$ activated\end{tabular}};
    \node[state] at (2,0) (2) {\begin{tabular}{@{} c @{} }$\tau_2$ waiting \\ urgent\end{tabular}};
    \node[state] at (0,-0.5) (3) {\begin{tabular}{@{} c @{} }$\tau_2$ activated\end{tabular}};
    \node[state] at (-2,-0.5) (4) {\begin{tabular}{@{} c @{} }$\mathcal{P}^1$ complete \\ $c_{\mathcal{P}^1} \leq T_1$\end{tabular}};
    
    \path
    (0)
    edge [above] node {$\tau_1$ activation} (1)
    
    (1)
    edge [above right] node {$\tau_1$ completed} (2)
    
    (2)
    edge [below right] node {$\tau_2$ activation} (3)
    
    (3)
    edge [below] node {$\tau_2$ completed} (4)
    
    (4)
    edge [right] node {\begin{tabular}{c}$c_{\mathcal{P}^1} == T_1$ \\ $\mathcal{P}^1$ restart\\ $c_{\mathcal{P}^1}:=0$  \end{tabular}} (0)
    ;
    
  \end{tikzpicture}
  
}

\caption{PTA modelling a pipeline $\mathcal{P}^1$ with two tasks $\tau_1,\tau_2$}
\label{fig:pipeline}
\end{figure}

In Figure \ref{fig:preemptive}, we present how we model a preemptive
processor. The processor can be \emph{idle}, waiting for a task
activations. As soon as a request has been received, it moves to one
of the states where the corresponding higher priority task is running.
If it receives another activation request, it moves to the state
corresponding to the highest priority task running. Moreover, while a
task executes, the scheduler automaton checks if the corresponding
pipeline misses its deadline. In the case of a deadline miss, the
processor moves to a special failure state and stops any further
computation.

\begin{figure}[ht!]
{
  \centering
  \tiny
  
  \tikzstyle{state}=[rectangle, rounded corners, minimum size=18pt, draw=black, text=black, inner sep=1.5pt,initial text={}]
  \tikzstyle{point}=[draw=none, fill=black, circle]
  
  \begin{tikzpicture}[scale = 1.8, auto, ->, thin] 
    
    \node[initial, state] at (-1.5,0) (0) {\begin{tabular}{@{} c @{} }Idle\\ $c_{\tau_1}$,$c_{\tau_4}$ stopped \end{tabular}};
    \node[state] at (0,1) (1) {\begin{tabular}{@{} c @{} }$\tau_1$ running\\ $c_{\tau_4}$ stopped \end{tabular}};
    \node[state] at (0,-1) (2) {\begin{tabular}{@{} c @{} }$\tau_4$ running \\ $c_{\tau_1}$ stopped \end{tabular}};
    \node[state] at (.8,0) (12) {\begin{tabular}{@{} c @{} }$\tau_1$ running\\$\tau_4$ activated\\$c_{\tau_4}$ stopped\end{tabular}};
    \node[state] at (2.3,0) (dm) {\begin{tabular}{@{} c @{} }Deadline missed\end{tabular}};
    
    \path
    (0)
    edge [bend left=60, above left] node {$\tau_1$ activation} (1)
    edge [bend left, below left] node {$\tau_4$ activation} (2)
    
    (1)
    edge [bend left, above left] node {\begin{tabular}{c}$c_{\tau_1} == C_1$ \\ $\tau_1$ completed\\ $c_{\tau_1}:=0$  \end{tabular}} (0)
    edge [above right] node {$\tau_4$ activation} (12)
    edge [bend left, above right] node {\begin{tabular}{c}$c_{\mathcal{P}^1} > D_{e2e}^1$ \\ Deadline miss \end{tabular}} (dm)
    
    (2)
    edge [bend left, below left] node {\begin{tabular}{c}$c_{\tau_4} == C_4$ \\ $\tau_4$ completed\\ $c_{\tau_4}:=0$  \end{tabular}} (0)
    edge [bend left, below right] node {$\tau_1$ activation} (12)
    edge [bend right,below right, out=-70,in=-90] node {\begin{tabular}{c}$c_{\mathcal{P}^2} > D_{e2e}^2$ \\ Deadline miss \end{tabular}} (dm)

    (12)
    edge [bend left, below right] node {\begin{tabular}{c}$c_{\tau_1} == C_1$ \\ $\tau_1$ completed\\ $c_{\tau_1}:=0$  \end{tabular}} (2)
    edge [below] node {\begin{tabular}{c}$c_{\mathcal{P}^1} > D_{e2e}^1 $ \\ or $c_{\mathcal{P}^2} > D_{e2e}^2$\\ Deadline miss  \end{tabular}} (dm)
    
    ;
    
  \end{tikzpicture}
  
}

\caption{PTA modelling a preemptive processor with two tasks $\tau_1,\tau_4$}
\label{fig:preemptive}
\end{figure}
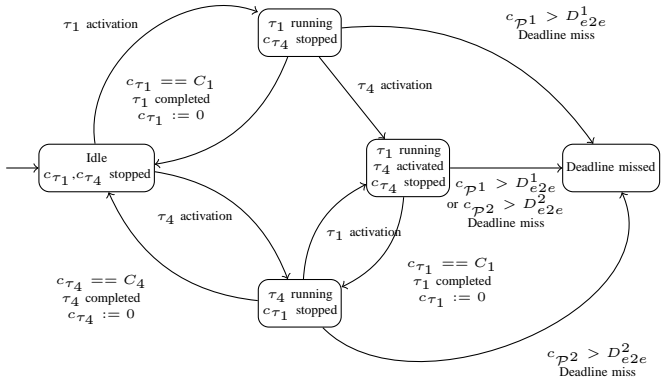

In Figure \ref{fig:nonpreemptive}, we present the model a
non-preemptive processor. Similarly to the previous case, the
processor can be idle, waiting for an activation request. As soon as a
request as been received it moves to a corresponding state, setting a
token corresponding to the activated task to $1$.  If another request
is sent at the same time, it sets a corresponding token to $1$, and
moves to the state where the highest priority task will be running.
Once a task is completed, the processor set the corresponding token to
$0$, and according to the token set to $1$, moves to the state where
the highest priority task will be running. Similarly to the previous
case, while a task executes, the automaton checks for deadline misses,
and in that case it stops any further computation by moving to a
special failure state.

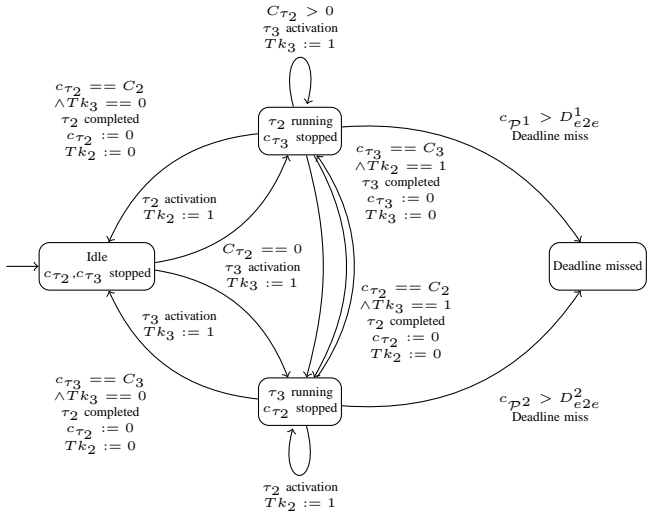
\begin{figure}[ht!]
{
  \centering
  \tiny
  
  \tikzstyle{state}=[rectangle, rounded corners, minimum size=18pt, draw=black, text=black, inner sep=1.5pt,initial text={}]
  \tikzstyle{point}=[draw=none, fill=black, circle]
  
  \begin{tikzpicture}[scale = 1.8, auto, ->, thin] 
    
    \node[initial, state] at (-2.5,0) (0) {\begin{tabular}{@{} c @{} }Idle\\ $c_{\tau_2}$,$c_{\tau_3}$ stopped \end{tabular}};
    \node[state] at (-1,1) (1) {\begin{tabular}{@{} c @{} }$\tau_2$ running\\ $c_{\tau_3}$ stopped \end{tabular}};
    \node[state] at (-1,-1) (2) {\begin{tabular}{@{} c @{} }$\tau_3$ running \\ $c_{\tau_2}$ stopped \end{tabular}};
    \node[state] at (1.2,0) (dm) {\begin{tabular}{@{} c @{} }Deadline missed\end{tabular}};
    
    \path
    (0)
    edge [bend right, above left] node {\begin{tabular}{c}$\tau_2$ activation\\$Tk_2:=1$ \end{tabular}} (1)
    edge [bend left, below left] node {\begin{tabular}{c}$\tau_3$ activation\\$Tk_3:=1$ \end{tabular}} (2)
    (1)
    edge [bend right, above left] node {\begin{tabular}{c}$c_{\tau_2} == C_2$\\ $\wedge Tk_3==0$ \\ $\tau_2$ completed\\ $c_{\tau_2}:=0$\\$Tk_2:=0$  \end{tabular}} (0)
    edge [loop above, above] node {\begin{tabular}{c}$C_{\tau_2} > 0$ \\ $\tau_3$ activation \\ $Tk_3:=1$ \end{tabular}} (1)
    edge [bend left = 15,left] node {\begin{tabular}{c}$C_{\tau_2} == 0$ \\ $\tau_3$ activation \\ $Tk_3:=1$ \end{tabular}} (2)
    edge [bend left = 30, right] node[near end] {\begin{tabular}{c}$c_{\tau_2} == C_2$\\ $\wedge Tk_3==1$ \\ $\tau_2$ completed\\ $c_{\tau_2}:=0$\\$Tk_2:=0$  \end{tabular}} (2)
    
    edge [bend left, above right] node {\begin{tabular}{c}$c_{\mathcal{P}^1} > D_{e2e}^1$ \\ Deadline miss \end{tabular}} (dm)
    (2)
    edge [bend left, below left] node {\begin{tabular}{c}$c_{\tau_3} == C_3$\\ $\wedge Tk_3==0$ \\ $\tau_2$ completed\\ $c_{\tau_2}:=0$\\$Tk_2:=0$  \end{tabular}} (0)
    edge [loop below, below] node {\begin{tabular}{c} $\tau_2$ activation \\ $Tk_2:=1$ \end{tabular}} (2)
    edge [bend right = 35, right] node[very near end] {\begin{tabular}{c}$c_{\tau_3} == C_3$\\ $\wedge Tk_2==1$ \\ $\tau_3$ completed\\ $c_{\tau_3}:=0$\\$Tk_3:=0$  \end{tabular}} (1)
    edge [bend right, below right] node {\begin{tabular}{c}$c_{\mathcal{P}^2} > D_{e2e}^2$ \\ Deadline miss \end{tabular}} (dm)


%
    ;
    
  \end{tikzpicture}
  
}

\caption{PTA modelling a non-preemptive processor with two tasks $\tau_2,\tau_3$}
\label{fig:nonpreemptive}
\end{figure}

Since we model periodic pipelines, and the model explores all possible
traces, we expect that schedulability region will be larger that the
one obtained with other techniques which only consider sporadic
pipelines (like the analysis proposed in Section
\ref{sec:analytic}). An assessment of this difference is provided in
the next section.

\section{Evaluation}
\label{sec:experiments}

We evaluated the effectiveness and the running time of three different
tools for parametric schedulability analysis: the RTSCAN tool, which
implements the analytic method described in Section
\ref{sec:analytic}; the IMITATOR tool \cite{FLMS-time12}, described in
Section \ref{sec:timed-automata}; and the MAST tool
\cite{MAST-web-page}.

We highlight that the three tools are implemented in different
languages, and use different ways to optimise the analysis; RTSCAN is
implemented in C/C++ and uses on the PPL library; IMITATOR is
implemented in OCaml, but it also used the PPL libraries for building
regions; finally, MAST is implemented in Ada. 

For MAST, we have selected the ``Offset Based analysis'', proposed in
\cite{Palencia1998}. For IMITATOR, we consider all pipelines as
strictly periodic, and only deadlines less than periods. We evaluated
the tools on two different test cases, in increasing order of
complexity. We will first present the results, in terms of
schedulability regions, for two different test cases. In order to
simplify the visualisation of the results, for each test case, we will
present the 2D region of two parameters only: however, all three
methods are general and can be applied to any number of parameters.

In Section \ref{sec:execution-times}, after discussing some important
implementation details, we will and present the execution times of the
three tools.

\subsection{Test case 1}

The first test case has been adapted from \cite{Palencia1998} (we
reduced the computation times of some tasks to position the system in
a interesting schedulability region). It consists of 2 processors,
connected by a CAN bus, three simple periodic tasks and one
pipeline. The parameters are listed in Table \ref{tab:tc1}. Processor
1 and 3 model two different computational nodes that are scheduled by
preemptive fixed priority, and Processor 2 models a CAN bus with
non-preemptive fixed priority policy. The only pipeline models a
remote procedure call from CPU 1 to CPU 3. All tasks have deadlines
equal to periods, and also the pipeline has end-to-end deadline equal
to its period. Only two messages are sent on the network, and if the
pipeline is schedulable, they cannot interfere with each other. We
wish to perform parametric schedulability analysis with respect to
$C_1$ and $C_1^1$.

\begin{table}
  \centering
  \begin{center}
\begin{tabular}{|c|c|c|c|c|c|c|} \hline
  Pipeline/Task    & $T$ & $D_{e2e}$ & Tasks      & $C$ & $\pi$ & $p$ \\\hline
  $\tau_1$         &  20 &   20     &   -        &  free  &   9   &  1  \\\hline
  \multirow{5}{*}{$P^1$} & \multirow{5}{*}{ 150 } & \multirow{5}{*}{ 150 } 
                                    & $\tau_1^1$ &  free  &   3   & 1   \\     
                   &     &          & $\tau_2^1$ &  10  &   9   & 2   \\
                   &     &          & $\tau_3^1$ &  8   &   5   & 3   \\
                   &     &          & $\tau_4^1$ &  15  &   2   & 2   \\
                   &     &          & $\tau_5^1$ &  25  &   2   & 1   \\
                   \hline
  $\tau_2$         &  30 &   30     &   -         &  6  &   9   &  3  \\\hline
  $\tau_3$         &  200 & 200     &   -         &  40 &   2   &  3  \\\hline                   
\end{tabular}
\end{center}
\caption{Test Case 1: one pipeline with deadline equal to period.}
\label{tab:tc1}
\end{table}

The resulting regions of schedulability from the tools RTSCAN and MAST
are reported in Figure \ref{fig:palencia-rtscan}, whereas the region
produced by IMITATOR is reported in Figure \ref{fig:palencia-imitator}
in green (the non schedulable region is also painted in red). The small
white triangles are regions which do not contain any integer point,
and have not been explored by the IMITATOR tool.

In this particular test, RTSCAN dominates MAST. After some debugging,
we discovered that the analysis algorithm currently implemented in
MAST does not consider the fact that the two messages $\tau_2^1$ and
$\tau_4^1$ cannot interfere with each other, and instead considers a
non-null blocking time on the network. This is probably a small bug in
the MAST implementation that we hope will be solved in a future
version.

Also, as expected, the region computed by IMITATOR dominates the other
two tools includes the other two regions. The reason is in the
different model of computation: IMITATOR considers fully periodic and
synchronous pipelines, therefore it produces all possible traces that
can be generated in this case. Both RTSCAN and MAST, instead, compute
upper bounds on the interference that a task can suffer from higher
priority task and from blocking time of lower priority
tasks. Therefore, they can only provide a sufficient analysis. 

\begin{figure}[t]
  \centering
  \psfrag{C11}[c][Bl][1][0]{$C_1^1$}
  \psfrag{C21}[c][Bl][1][180]{$C_1$}
  \includegraphics[width=\columnwidth]{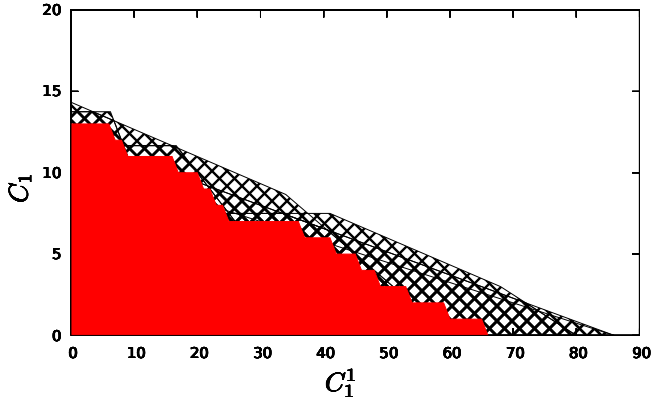}
  \caption{Schedulability regions for test case 1, produced by RTSCAN
    (hatched pattern) and MAST (filled pattern)}
  \label{fig:palencia-rtscan}
\end{figure}

\begin{figure}[h]
  \centering
  \psfrag{C\_task\_11}[c][Bl][1][180]{$C_1$}
  \psfrag{C\_task\_21}{$C_1^1$}
  \psfrag{test.imi}{~}
  \includegraphics[width=.75\columnwidth, height=0.25\textheight]{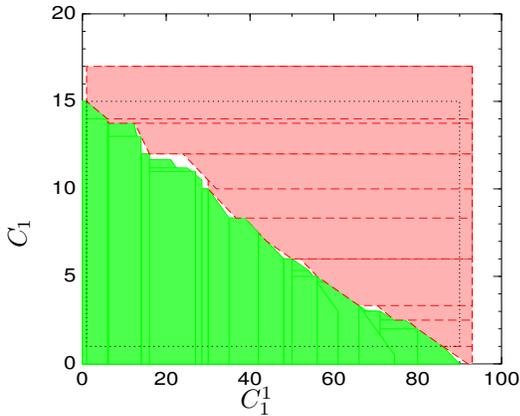}
  \caption{Schedulability regions for test case 1, produced by IMITATOR (green (lower) region)}
  \label{fig:palencia-imitator}
\end{figure}

\subsection{Test case 2}

The second test case is taken from
\cite{Wandeler:2006:SAE:1177177.1177184}. It consists of two pipelines
on 3 processors (with id 1, 3 and 4) and one network (with id 2).  We
actually consider two versions of this test case: in the first version
(a) pipeline $P^1$ is periodic with period $200msec$ and end-to-end
deadline equal to the period. In the second version (b), the period of
the first pipeline is reduced to $30 msec$ 
(as in the original specification in
\cite{Wandeler:2006:SAE:1177177.1177184}).  The full set of parameters
is reported in Table~\ref{tab:tc2}, where all values are expressed in
microseconds. We perform parametric analysis on $C_5^1$ and $C_1^2$.

\begin{table}
  \centering
  \begin{center}
\begin{tabular}{|c|c|c|c|c|c|c|c|} \hline
  Pipeline         & $T$ & $D_{e2e}$ & Tasks      & $C$     & $\pi$ & $p$ \\\hline
  \multirow{5}{*}{$P^1$} & \multirow{5}{1cm}{ 200,000 (30,000) } & \multirow{5}{*}{ 200,000 } 
                                    & $\tau_1^1$ &  4,546   &   10  & 1    \\     
                   &     &          & $\tau_2^1$ &  445     &   10  & 2    \\
                   &     &          & $\tau_3^1$ &  9,091   &   10  & 4    \\
                   &     &          & $\tau_4^1$ &  445     &   9   & 2    \\
                   &     &          & $\tau_5^1$ &  free    &   9   & 1    \\
                   \hline
  \multirow{5}{*}{$P^2$} & \multirow{5}{1cm}{ 300,000 } & \multirow{5}{*}{ 100,000 } 
                                    & $\tau_1^2$ &  free     &   9   & 4    \\     
                   &     &          & $\tau_2^2$ &  889     &   8   & 2    \\
                   &     &          & $\tau_3^2$ &  44,248  &   10  & 3    \\
                   &     &          & $\tau_4^2$ &  889     &   7   & 2    \\
                   &     &          & $\tau_5^2$ &  22,728  &   8   & 1    \\
                   \hline
\end{tabular}
\end{center}

\caption{Test Case 2: two pipelines on 3 processors}
\label{tab:tc2}
\end{table}

For version (a) we run all tools and we report the regions of
schedulability in Figure \ref{fig:thiele-a-rtscan} for RTSCAN and
MAST. In this case, MAST dominated RTSCAN. The reason is due to the
offset based analysis methodology used in MAST, which reduces the
interference on one task from other tasks belonging to the same
pipeline. RTSCAN does not implement such an optimisation (it will be
the topic of future extensions) and hence it is more pessimistic.

\begin{figure}[t]
  \centering
  \psfrag{C21}[c][Bl][1][180]{$C_5^1$}
  \psfrag{C15}[c][Bl][1][0]{$C_1^2$}
  \includegraphics[width=\columnwidth]{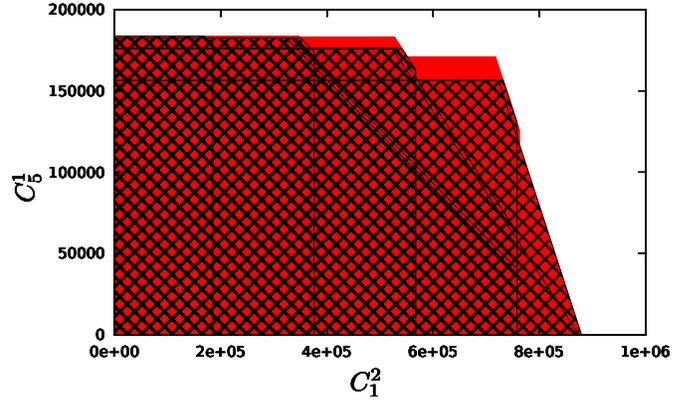}
  \caption{Schedulability regions for test case 2a, produced by RTSCAN
    (hatched) and MAST (filled)}
  \label{fig:thiele-a-rtscan}
\end{figure}

The results from IMITATOR are shown in Figure
\ref{fig:thiele-a-imitator}. Again, the region produced by IMITATOR
dominates the one produced by the other two tools.

\begin{figure}[b]
  \centering
  \psfrag{C\_task\_21}{$C_1^2$}
  \psfrag{C\_task\_15}[c][Bl][1][180]{$C_5^1$}
  \psfrag{test.imi}{~}
  \includegraphics[width=.75\columnwidth]{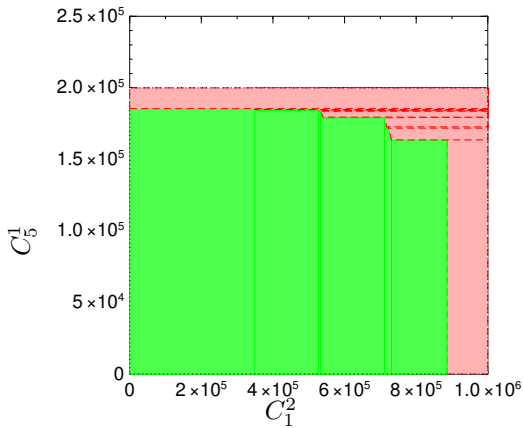}
  \caption{Schedulability regions for test case 2a, produced by IMITATOR (green region)}
  \label{fig:thiele-a-imitator}
\end{figure}

For version (b) we run only RTSCAN and MAST, because in the current
version of IMITATOR we can only model constrained deadline
systems. The results for version (b) are reported in Figure
\ref{fig:thiele-b-rtscan}. In this case, MAST dominates RTSCAN. Again,
this is probably due to the fact that MAST implements the offset-based
analysis.

\begin{figure}[t]
  \centering
  \psfrag{XA}[c][Bl][1][0]{$C_1^2$}
  \psfrag{MYCYA}[c][Bl][1][0]{$C_5^1$}
  \psfrag{test.imi}{~}
  \includegraphics[width=\columnwidth]{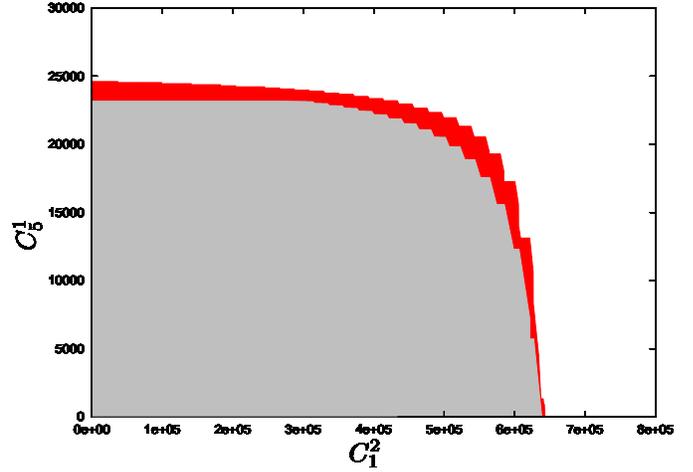}
  \caption{Schedulability regions for test case 2b, produced by RTSCAN
    (grey filled) and MAST (red filled)}
  \label{fig:thiele-b-rtscan}
\end{figure}

\subsection{Execution times}
\label{sec:execution-times}

Before looking at the execution times of the three tools in the three
different test cases, it is worth to discuss some detail about their
implementation. First of all, all the three tools are single threaded,
therefore we did not use any kind of parallelisation technique in
order to have a fair comparison. The RTSCAN tool uses the technique
described in Section \ref{sec:analytic}, and as a result it produces a
disjunction of convex regions, each one corresponds to a set of
inequalities in AND. Typically, the number of convex regions produced
by the tool is relatively small. Also, there is no need to ``explore''
the space of feasible points, as these regions are naturally obtained
from the problem constraints.

IMITATOR also produces a disjunction of convex regions. However, these
regions are typically smaller and disjoints. Moreover, to produce a
region, IMITATOR needs to start from a candidate point on which to
perform a sensitivity analysis. More specifically, it works as
follows: 1) it starts from a random point (typically the centre of the
interval) and computes a region around it. Then, it searches for the
next candidate point outside of the already found regions. The key
factor here is how this search is performed. Currently, IMITATOR
search for a candidate point in the neighbourhood of the current
region. This is a very general strategy that works for any kind of
PTA. However, the particular structure of schedulability problems
would probably require an ad-hoc exploration algorithm.

MAST can perform schedulability analysis given a full set of parameter
values, and it returns either a positive or a negative response. In
addition, MAST can perform sensitivity analysis on one parameter
(called \emph{slack computation} in the tool), using binary search on
a possible interval of values. This latter strategy can be used to
implement parametric analysis: we select a interval of values for each
free parameter that we wish to analyse. Then, we perform a cycle on
all values of one parameter (with a predefined step) and we ask MAST
to compute the interval of feasible values for the other parameter.

This may not be the smartest way to proceed: it is possible, for
example, to implement binary search on the full 2D space of free
parameters to accelerate the execution time of the tool. We defer the
implementation of such an algorithm as a future extension.

All experiments have been performed onto a PC with 8Gb of RAM, an Intel
Core I7 quad-core processor, working at 800 Mhz per processor. 

We are now ready to present and discuss the execution times of the
tools in the three test cases, which are reported in Table
\ref{tab:exec}, together with the length of the hyperperiod for the
test cases. As you can see, for small problems RTSCAN performs very
well.  In test case 2b, the execution time of RTSCAN is much larger
than the one obtained from test case 2a. This is due to the fact that
in test case 2b, one pipeline has end-to-end deadline greater than the
period, and therefore RTSCAN needs to compute many more inequalities
(for all points in the hyperperiod).

As for MAST, in test cases 2a and 2b (where the time units are
expressed in microseconds, and therefore are quite large), the search
has been run with a step of 100 for a good compromise between
precision and execution time. However, we believe that a smarter
algorithm for exploring the space of parameters can really improve the
overall execution time.

Finally, IMITATOR greatly suffers from a similar problem. We observed
that the tool spends a few seconds for computing the schedulability
region around each point. However, the regions are quite small, and
there are many of them: for example, in test case 2a IMITATOR analysed
257 regions. Also, it spends a large amount of time in searching for
neighbourhood points. Therefore, we believe that a huge improvement in
the computation time of IMITATOR can be achieved by coming up with a
smarter way of exploring the schedulability space.

\begin{table}[h]
  \centering
  \begin{tabular}{|c|c|c|c|c|}\hline
    Test Case & Hyperperiod   &  RTSCAN   &    MAST     &    IMITATOR   \\\hline
        1     &    600       &  0.27s     &    7.19 s   &     19m42     \\
        2a    &    600,000   &  0.47s     &    40m13s   &     4h        \\
         2b    &    300,000 &  1 m 47s   &    33m19s   &        --     \\
        \hline
  \end{tabular}
  \caption{Execution times of the tools in the three test cases}
  \label{tab:exec}
\end{table}

\section{Conclusions and Future Works}
\label{sec:conclusions}

In this paper we presented two different approaches to perform
parametric analysis of distributed real-time systems: one based on
analytic methods of classic schedulability analysis; the other one
based on model checking of PTA. We compared the two approached with
classical holistic analysis.

The results are promising, and we plan to extend this work along
different directions. Regarding the analytic method, we want to
enhance the analysis including static and dynamic offsets, following
the approach of \cite{Palencia1998}. Also, in the future we will to
extend the model to consider mutually exclusive semaphores and
multiprocessor scheduling.

Regarding IMITATOR, we plan to improve the algorithm to explore the
space of parameters: one promising idea is to use the analytic method
to find an initial approximation of the feasible space, and then
extend the border of the space using PTAs. We also plan to collaborate
with the team at Universidad de Cantabria that develops the MAST tool
on novel algorithms for exploring an N-dimensions parameter space.

\section{Acknowledgements}
\label{sec:acks}

We would like to express our gratitude to Michael González Harbour and
Juan M. Rivas, from the Universidad de Cantabria, for their support to
installing and using the MAST tool.

The research leading to these results has received funding from the
European Union Seventh Framework Programme (FP7/2007-2013) under grant
agreement No. 246556.

\balance

\bibliographystyle{IEEEtran}
\bibliography{all}

\end{document}